\documentclass[12pt]{article}

\usepackage{array, amssymb, amsmath, graphics}
\newtheorem{theorem}{Theorem}
\newtheorem{lemma}{Lemma}

\newtheorem{proposition}{Proposition}

\newtheorem{corollary}{Corollary}
\newtheorem{question}{Question}
\newtheorem{conjecture}{Conjecture}

\newenvironment{proof}{{\bf Proof.}}{\hfill\rule{2mm}{2mm}}
\newtheorem{remarka}{Remark}

\def\emline#1#2#3#4#5#6{%
       \put(#1,#2){\special{em:moveto}}%
       \put(#4,#5){\special{em:lineto}}}
\def\newpic#1{}

\def\so{{\rm so}}

\newlength{\cellwid}
\makeatletter
{\endarray}
\title{On minimum vertex cover of \\generalized Petersen graphs}

\author
{%
Babak Behsaz\thanks{Department of Computer Engineering and
Information Technology, Amirkabir University of Technology (Tehran
Polytechnic), 424 Hafez Ave., Tehran, Iran. Email:\ {\tt
behsaz@ce.aut.ac.ir.}}
\and
Pooya Hatami\thanks{Department of Mathematical Sciences, Sharif
University of Technology, P.O. Box 11365-9415, Azadi Ave., Tehran,
Iran. Email:\ {\tt p\_hatami@ce.sharif.edu.}}
\and
Ebadollah S. Mahmoodian\thanks{Department of Mathematical Sciences,
Sharif University of Technology, P.O. Box 11365-9415, Azadi Ave.,
Tehran, Iran. Email:\ {\tt emahmood@sharif.edu.}}
}%

\date{}

\begin{document}

\maketitle

\begin{abstract}
For natural numbers $n$ and $k$ ($n > 2k$), a generalized Petersen
graph $P(n,k)$, is defined by vertex set $\lbrace u_i,v_i\rbrace$
and edge set $\lbrace u_iu_{i+1},u_iv_i,v_iv_{i+k}\rbrace$; where
$i = 1,2,\dots,n$ and subscripts are reduced modulo $n$. Here
first, we characterize minimum vertex covers in generalized
Petersen graphs. Second, we present a lower bound and some upper
bounds for $\beta(P(n,k))$, the size of minimum vertex cover of
$P(n,k)$. Third, in some cases, we determine the exact values of
$\beta(P(n,k))$. Our conjecture is that $\beta(P(n,k)) \le n +
\lceil\frac{n}{5}\rceil$, for all $n$ and $k$.
\end{abstract}

\section{Introduction and preliminaries}
For the definition of basic concepts not given here, one may refer
to a text book in graph theory, for example~\cite{west01}. A set $Q$
of vertices of a graph $G = (V,E)$ is called a \textsf{vertex
cover}, if each edge in $E$ has at least one endpoint in $Q$. A
vertex cover with minimum size in a graph $G$, is called a
\textsf{minimum vertex cover} of $G$ and its size is denoted by
$\beta(G)$. It is well-known that the vertex-cover problem is an
NP-complete problem~\cite[p.~1006]{cormen01}. Therefore, many
attempts are made to find lower and upper bounds, and exact values
of $\beta(G)$ for special classes of graphs.

This paper is aimed toward studying vertex cover problem for the
class of generalized Petersen graphs. In a generalized Petersen
graph $P(n,k)$, as defined in abstract, let $U = \{u_1, u_2,
\dots, u_n\}$ and $V = \{v_1, v_2, \dots, v_n\}$. We call two
vertices $u_i$ and $v_i$ as \textsf{twin} of each other and the
edge between them as a \textsf{spoke}. Also, by twin of $S$,
where $S$ is a subset of $V$, we mean the set that contains twins
of all members of $S$. Moreover, the edges with both endpoints in
$U$ are called \textsf{$U$-edges} and the edges with both
endpoints in $V$ are called \textsf{$V$-edges}. A maximal subset
of consecutive vertices of $V$ in a vertex cover $Q$, is called a
\textsf{strip} of $Q$ (two vertices of $V$ are consecutive if
they have circular consecutive subscripts). We define the
\textsf{size of a strip} by the number of its vertices. Note that
the size of a strip may be equal to $1$. Also, we call a strip
\textsf{odd} if it has an odd size. Finally, we call the set of
$m$ consecutive twins, namely $\lbrace
u_{i+1},u_{i+2},\dots,u_{i+m} \rbrace \thinspace \cup \thinspace
\lbrace v_{i+1},v_{i+2},\dots,v_{i+m} \rbrace$ an
\textsf{$m$-sector} of generalized Petersen graph. In
Figure~\ref{fig1} a vertex cover is shown for $P(16,5)$. For
example, the set $\{v_{16},v_1,v_2,v_3,v_4\}$ is a strip.

\vspace*{.5cm}
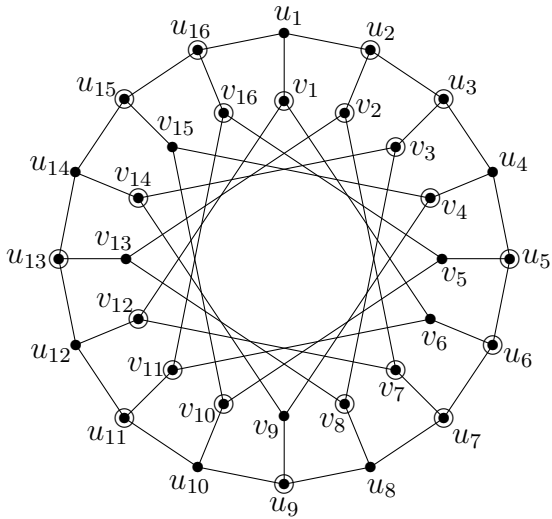
\begin{figure}[ht]
\label{fig1}
\unitlength 0.6mm
\special{em:linewidth 0.4pt}
\linethickness{0.4pt}

\begin{center}

\begin{picture}(100,105)(-47,-55)
\put(0,35){\circle*{2.5}} \put(0,35){\circle{4}}
\put(0,50){\circle*{2.5}}
\put(4.87196,37.7918){\makebox(0,0)[cc]{$v_{1}$}}
\put(1,54){\makebox(0,0)[cc]{$u_{1}$}}
\put(13.3935,32.3359){\circle*{2.5}}\put(13.3935,32.3359){\circle{4}}
\put(19.1336,46.1942){\circle*{2.5}}\put(19.1336,46.1942){\circle{4}}
\put(18.8315,33.3953){\makebox(0,0)[cc]{$v_{2}$}}
\put(21.4643,50.4597){\makebox(0,0)[cc]{$u_{2}$}}
\put(24.7482,24.7493){\circle*{2.5}}\put(24.7482,24.7493){\circle{4}}
\put(35.3545,35.3562){\circle*{2.5}}\put(35.3545,35.3562){\circle{4}}
\put(30.8309,23.915){\makebox(0,0)[cc]{$v_{3}$}}
\put(39.5829,38.5847){\makebox(0,0)[cc]{$u_{3}$}}
\put(32.3353,13.395){\circle*{2.5}}\put(32.3353,13.395){\circle{4}}
\put(46.1933,19.1358){\circle*{2.5}}
\put(37.8347,10.794){\makebox(0,0)[cc]{$v_{4}$}}
\put(50.8888,21.5666){\makebox(0,0)[cc]{$u_{4}$}}
\put(35,0.00162144){\circle*{2.5}}
\put(50,0.00231634){\circle*{2.5}}\put(50,0.00231634){\circle{4}}
\put(37.792,-3.97021){\makebox(0,0)[cc]{$v_{5}$}}
\put(56,0.00250165){\makebox(0,0)[cc]{$u_{5}$}}
\put(32.3366,-13.392){\circle*{2.5}}
\put(46.1951,-19.1315){\circle*{2.5}}\put(46.1951,-19.1315){\circle{4}}
\put(33.3962,-18.13){\makebox(0,0)[cc]{$v_{6}$}}
\put(51.8907,-21.962){\makebox(0,0)[cc]{$u_{6}$}}
\put(24.7505,-24.747){\circle*{2.5}}\put(24.7505,-24.747){\circle{4}}
\put(35.3578,-35.3529){\circle*{2.5}}\put(35.3578,-35.3529){\circle{4}}
\put(23.9163,-29.5298){\makebox(0,0)[cc]{$v_{7}$}}
\put(40.4864,-39.8811){\makebox(0,0)[cc]{$u_{7}$}}
\put(13.3965,-32.3347){\circle*{2.5}}\put(13.3965,-32.3347){\circle{4}}
\put(19.1379,-46.1924){\circle*{2.5}}
\put(10.7956,-36.4342){\makebox(0,0)[cc]{$v_{8}$}}
\put(21.6689,-50.8878){\makebox(0,0)[cc]{$u_{8}$}}
\put(0.00324288,-35){\circle*{2.5}}
\put(0.00463268,-50){\circle*{2.5}}\put(0.00463268,-50){\circle{4}}
\put(-3.96846,-37.7922){\makebox(0,0)[cc]{$v_{9}$}}
\put(0.00500329,-55){\makebox(0,0)[cc]{$u_{9}$}}
\put(-13.3905,-32.3372){\circle*{2.5}}\put(-13.3905,-32.3372){\circle{4}}
\put(-19.1294,-46.196){\circle*{2.5}}
\put(-19.1284,-34.097){\makebox(0,0)[cc]{$v_{10}$}}
\put(-20.9597,-49.9916){\makebox(0,0)[cc]{$u_{10}$}}
\put(-24.7459,-24.7516){\circle*{2.5}}\put(-24.7459,-24.7516){\circle{4}}
\put(-35.3512,-35.3594){\circle*{2.5}}\put(-35.3512,-35.3594){\circle{4}}
\put(-30.5287,-23.9177){\makebox(0,0)[cc]{$v_{11}$}}
\put(-39.1793,-39.1882){\makebox(0,0)[cc]{$u_{11}$}}
\put(-32.3341,-13.398){\circle*{2.5}}\put(-32.3341,-13.398){\circle{4}}
\put(-46.1915,-19.1401){\circle*{2.5}}
\put(-37.4337,-9.7973){\makebox(0,0)[cc]{$v_{12}$}}
\put(-51.6869,-20.9713){\makebox(0,0)[cc]{$u_{12}$}}
\put(-35,-0.00486431){\circle*{2.5}}
\put(-50,-0.00694902){\circle*{2.5}}\put(-50,-0.00694902){\circle{4}}
\put(-37.7924,3.96671){\makebox(0,0)[cc]{$v_{13}$}}
\put(-57,-0.00750494){\makebox(0,0)[cc]{$u_{13}$}}
\put(-32.3378,13.3891){\circle*{2.5}}\put(-32.3378,13.3891){\circle{4}}
\put(-46.1969,19.1272){\circle*{2.5}}
\put(-33.3978,18.1269){\makebox(0,0)[cc]{$v_{14}$}}
\put(-51.8926,20.6574){\makebox(0,0)[cc]{$u_{14}$}}
\put(-24.7527,24.7447){\circle*{2.5}}
\put(-35.3611,35.3496){\circle*{2.5}}\put(-35.3611,35.3496){\circle{4}}
\put(-23.9191,29.5276){\makebox(0,0)[cc]{$v_{15}$}}
\put(-41.50,38.1776){\makebox(0,0)[cc]{$u_{15}$}}
\put(-13.3995,32.3335){\circle*{2.5}}\put(-13.3995,32.3335){\circle{4}}
\put(-19.1422,46.1907){\circle*{2.5}}\put(-19.1422,46.1907){\circle{4}}
\put(-9.799,36.4332){\makebox(0,0)[cc]{$v_{16}$}}
\put(-20.6736,50.8859){\makebox(0,0)[cc]{$u_{16}$}}
\emline{0}{50}{1}{19.1336}{46.1942}{2}
\emline{0}{35}{1}{32.3366}{-13.392}{2} \emline{0}{35}{1}{0}{50}{2}
\emline{19.1336}{46.1942}{1}{35.3545}{35.3562}{2}
\emline{13.3935}{32.3359}{1}{24.7505}{-24.747}{2}
\emline{13.3935}{32.3359}{1}{19.1336}{46.1942}{2}
\emline{35.3545}{35.3562}{1}{46.1933}{19.1358}{2}
\emline{24.7482}{24.7493}{1}{13.3965}{-32.3347}{2}
\emline{24.7482}{24.7493}{1}{35.3545}{35.3562}{2}
\emline{46.1933}{19.1358}{1}{50}{0.00231634}{2}
\emline{32.3353}{13.395}{1}{0.00324288}{-35}{2}
\emline{32.3353}{13.395}{1}{46.1933}{19.1358}{2}
\emline{50}{0.00231634}{1}{46.1951}{-19.1315}{2}
\emline{35}{0.00162144}{1}{-13.3905}{-32.3372}{2}
\emline{35}{0.00162144}{1}{50}{0.00231634}{2}
\emline{46.1951}{-19.1315}{1}{35.3578}{-35.3529}{2}
\emline{32.3366}{-13.392}{1}{-24.7459}{-24.7516}{2}
\emline{32.3366}{-13.392}{1}{46.1951}{-19.1315}{2}
\emline{35.3578}{-35.3529}{1}{19.1379}{-46.1924}{2}
\emline{24.7505}{-24.747}{1}{-32.3341}{-13.398}{2}
\emline{24.7505}{-24.747}{1}{35.3578}{-35.3529}{2}
\emline{19.1379}{-46.1924}{1}{0.00463268}{-50}{2}
\emline{13.3965}{-32.3347}{1}{-35}{-0.00486431}{2}
\emline{13.3965}{-32.3347}{1}{19.1379}{-46.1924}{2}
\emline{0.00463268}{-50}{1}{-19.1294}{-46.196}{2}
\emline{0.00324288}{-35}{1}{-32.3378}{13.3891}{2}
\emline{0.00324288}{-35}{1}{0.00463268}{-50}{2}
\emline{-19.1294}{-46.196}{1}{-35.3512}{-35.3594}{2}
\emline{-13.3905}{-32.3372}{1}{-24.7527}{24.7447}{2}
\emline{-13.3905}{-32.3372}{1}{-19.1294}{-46.196}{2}
\emline{-35.3512}{-35.3594}{1}{-46.1915}{-19.1401}{2}
\emline{-24.7459}{-24.7516}{1}{-13.3995}{32.3335}{2}
\emline{-24.7459}{-24.7516}{1}{-35.3512}{-35.3594}{2}
\emline{-46.1915}{-19.1401}{1}{-50}{-0.00694902}{2}
\emline{-32.3341}{-13.398}{1}{0}{35}{2}
\emline{-32.3341}{-13.398}{1}{-46.1915}{-19.1401}{2}
\emline{-50}{-0.00694902}{1}{-46.1969}{19.1272}{2}
\emline{-35}{-0.00486431}{1}{13.3935}{32.3359}{2}
\emline{-35}{-0.00486431}{1}{-50}{-0.00694902}{2}
\emline{-46.1969}{19.1272}{1}{-35.3611}{35.3496}{2}
\emline{-32.3378}{13.3891}{1}{24.7482}{24.7493}{2}
\emline{-32.3378}{13.3891}{1}{-46.1969}{19.1272}{2}
\emline{-35.3611}{35.3496}{1}{-19.1422}{46.1907}{2}
\emline{-24.7527}{24.7447}{1}{32.3353}{13.395}{2}
\emline{-24.7527}{24.7447}{1}{-35.3611}{35.3496}{2}
\emline{-19.1422}{46.1907}{1}{0}{50}{2}
\emline{-13.3995}{32.3335}{1}{35}{0.00162144}{2}
\emline{-13.3995}{32.3335}{1}{-19.1422}{46.1907}{2}
\end{picture}

\caption{A vertex cover of P(16,5).}
\end{center}

\end{figure}

We now introduce some more definitions and conventions that will
be used in the following sections. First, \textsf{from this point
on} for brevity, we use \textsf{cover} instead of \textsf{vertex
cover}. Second, we show covers by lower case alphabet, because we
consider them as a function on $U\cup V$, which defines the
selected and unselected vertices. In addition, by $\lvert
c\rvert$ we mean the number of vertices selected by a cover $c$.
Third, we adopt the convention that all subscripts of vertices of
a generalized Petersen graph when they are out of the range
$\{1,2,\dots,n\}$, to be reduced modulo $n$ to belong to this set.
Fourth, in any generalized Petersen graph a cover which selects
all $V$-vertices, besides some $U$-vertices, is called a {\sf
trivial cover}, and it is not a minimum cover.  We denote by
$\mathcal{C}(P(n,k))$ the set of all non-trivial covers of
$P(n,k)$. Without loss of correctness {\it from this point on} by
a cover we mean a non-trivial cover. This assumption is needed
for Lemma~\ref{lem1}, part (a).

It seems that Watkins~\cite{watkins69} was the first who introduced
this class of generalized Petersen graphs and conjectured that they
have a Tait coloring, except $P(5,2)$. This conjecture later was
proved in~\cite{castagna72}. Since then this class of graphs has
been studied widely because of its interesting traits. There are
papers discussing topics such as tough sets, labeling problems, wide
diameters, and coloring of generalized Petersen graphs. For example
papers~\cite{alspach83},~\cite{bannai78}~and~\cite{schwenk89} are
about hamiltonian character of generalized Petersen graphs. Also
crossing numbers of this class is studied in papers such as~
\cite{mcquillan92} and~\cite{salazar05}. Recently vertex domination
of generalized Petersen graphs is studied~\cite{ebrahimi06}.

This paper is organized as follows. In Section~$2$ we characterize
properties of minimum vertex cover of generalized Petersen graphs by
defining a new quantity. In Section $3$ we introduce a lower bound
and some upper bounds for various cases of these graphs, namely:
\begin{itemize}
  \item $\beta(P(n,k)) \ge n+\frac{(n,k)+1}{2}$, for all odd $n$,
where $(n,k)$ is the greatest common divisor of $n$ and $k$;
  \item $\beta(P(n,k))\le n + \frac{k+1}{2}$ for all odd $k$;
  \item for all $m<k$, \\$\beta(P(n,k))\le \frac{n}{m} \beta(P(m,k\pmod
  m))$,
  where $m\mid n$ \ and \\  $\beta(P(n,k))\le \lfloor\frac{n}{m} \rfloor
  \beta(P(m,k \pmod m))+2k$, otherwise.
\end{itemize}
We note that most of these bounds are sharp. In Section~$4$, we
determine the exact values of $\beta(P(n,k))$:
\begin{itemize}
    \item for $k=1$ and $3$;
    \item when $n$ is even and $k$ is odd;
    \item when both $n$ and $k$ are odd and $k\mid n$.
\end{itemize}
Section $5$ contains concluding remarks.

\section{Properties of minimum vertex cover}
Let $S=\lbrace v_i,v_{i+1},\dots,v_{i+m}\rbrace$ be a strip of a
cover $c$. By \textsf{optimally selecting from twins of} $S$, we
mean choosing from $u_i,u_{i+1},\dots,u_{i+m}$ alternatively
beginning with $u_{i+1}$. For example, the cover in
Figure~\ref{fig1} is optimally selected from twins of the strip
$\lbrace v_{10},v_{11},v_{12}\rbrace$, but not from twins of the
strip $\lbrace v_{16},v_1,v_2,v_3,v_4\rbrace$. From any  cover
$c$ we construct a \textsf{semi-optimal} cover, $\so(c)$, in the
following way: First, we choose all vertices from $V$ selected by
$c$. Second, we choose twins of all vertices of $V$ which are not
selected by $c$. At last, we optimally select from twins of
strips of $c$. In Lemma~\ref{lem1}, we show that $\so(c)$ is a
cover, whose size is not greater than $c$.

\begin{lemma}\label{lem1}
Let $c$ be a cover of the generalized Petersen graph $P(n,k)$,
then
\begin{enumerate}
    \item[(a)] $\so(c)$ is a cover of $P(n,k)$, and
    \item[(b)] $\lvert \so(c)\rvert\le\lvert c\rvert$.
\end{enumerate}
\end{lemma}
\begin{proof}
(a) Note that the vertices selected from $V$ are the same in $c$
and $\so(c)$, so all $V$-edges are covered by $\so(c)$. Also, it
is evident that all spokes are covered by $\so(c)$. We show that,
all $U$-edges are also covered by $\so(c)$. For, suppose in
contradiction an edge $u_iu_{i+1}$ is not covered by $\so(c)$.
Then, the vertices $v_i$ and $v_{i+1}$ must have been selected by
$c$ and are in the same strip, say $S$. Since twins of $S$ must
be optimally selected, at least one of $u_i$ or $u_{i+1}$ must be
selected in $\so(c)$ (Note that as stated before, we assumed that
in the covers which we consider, not all vertices of $V$ are
selected). This is in contradiction with our assumption.

(b) Note that the only possible difference between $c$ and $\so(c)$
are in their selection from twins of their strips. It is enough to
show that for any strip $S=\{v_i,v_{i+1},\dots,v_{i+m-1}\}$ of $c$
and $\so(c)$, the number of selected vertices from twins of $S$ in
$\so(c)$ is not greater than the number of selected vertices from
twins of $S$ in $c$. To prove this, first, note that $\so(c)$
selects exactly $\lceil \frac{m-1}{2} \rceil$ vertices from the
twins of $S$. Second, note that $m-1$ edges between
$u_i,u_{i+1},\dots,u_{i+m-1}$ must be covered by the twins of $S$.
Since each vertex can cover at most two $U$-edges, the minimum
number of vertices that must be chosen from the twins of $S$ in any
cover is $\lceil \frac{m-1}{2} \rceil$. This shows that $\so(c)$
chooses the minimum possible vertices from the twins of $S$.
\end{proof}
\begin{corollary}\label{cor:new}
For a minimum cover $c^*$, we have $\lvert so(c^*) \rvert = \lvert
c^* \rvert$.
\end{corollary}

Now we are ready to present our main characterizing theorem. Let
$a(c)$ be the number of vertices selected from $V$ by a cover $c$ of
a generalized Petersen graph $P(n,k)$, and let $b(c)$ denote the
number of odd strips of $c$. Also, let $d(P(n,k)) = \min_{c \in
\mathcal{C}(P(n,k))}(a(c)-b(c))$.
\begin{theorem}\label{thm1}
For any minimum vertex cover $c^*$ of $P(n,k)$ we have,
$a(c^*)-b(c^*) = d(P(n,k))$.
\end{theorem}
\begin{proof}
Consider an arbitrary cover $c$ of $P(n,k)$. Suppose there are $s$
strips in $c$ with the sizes $m_1,m_2,\dots,m_s$. Without loss of
generality, we assume the first $b(c)$ strips in this list are all
odd. In $\so(c)$, there are $a(c)=\sum_{i=1}^s m_i$ vertices
selected from $V$ and $n-a(c)+\sum_{i=1}^s \lceil
\frac{m_i-1}{2}\rceil$ vertices selected from $U$ (the last term
of this expression shows the number of vertices selected from the
twins of strips of $c$) are selected in $\so(c)$. Thus, we obtain
\begin{align}
\lvert \so(c)\rvert&=a(c)+n-a(c)+\sum_{i=1}^s
\left\lceil\cfrac{m_i-1}{2}\right\rceil=n+\sum_{i=1}^s
\left\lceil\cfrac{m_i-1}{2}\right\rceil\notag\\
&=n+\sum_{i=1}^{b(c)}
\left\lceil\cfrac{m_i-1}{2}\right\rceil+\sum_{i=b(c)+1}^s
\left\lceil\cfrac{m_i-1}{2}\right\rceil=n+\sum_{i=1}^{b(c)}
\cfrac{m_i-1}{2}+\sum_{i=b(c)+1}^s\cfrac{m_i}{2}\notag\\
\label{socardi}&=n+\sum_{i=1}^s \cfrac{m_i}{2}-\cfrac{b(c)}{2}=n+
\cfrac{a(c)-b(c)}{2}.
\end{align}
Now, take the contradictory assumption that $d(P(n,k)) <
a(c^*)-b(c^*)$. Clearly, a cover exists that
$a(c)-b(c)=d(P(n,k))$. Therefore, by applying
Equation~\eqref{socardi} we have,
\begin{align*}
a(c)-b(c)<a(c^*)-b(c^*)&\Rightarrow
n+\cfrac{a(c)-b(c)}{2}<n+\cfrac{a(c^*)-b(c^*)}{2}\\
&\Rightarrow\lvert \so(c)\rvert < \lvert \so(c^*)\rvert.
\end{align*}
From Corollary~\ref{cor:new}, we reach the contradictory result that
$\lvert \so(c)\rvert < \lvert c^*\rvert$ where $so(c)$ is a cover by
Lemma~\ref{lem1}(a). By the definition of $d(P(n,k))$, $d(P(n,k))\le
a(c^*)-b(c^*)$. Thus, for every minimum cover $c^*$, we must have
$a(c^*)-b(c^*)=d(P(n,k))$.
\end{proof}

\begin{remarka}
Note that the converse of Theorem~\ref{thm1} is not necessarily true
and a cover c with $a(c) - b(c) = d(P(n,k))$ may not be a minimum
one. This fact is clear, because $d(P(n,k))$ imposes a constraint
only for the selections from $V$ and not on all vertices.
\end{remarka}

Now we use semi-optimal covers to characterize minimum covers in
generalized Petersen graphs.
\begin{proposition}\label{thm2}
In a minimum cover $c^*$ of $P(n,k)$, the maximum size of a strip is
at most $2k$ when $k$ is odd, and $2k+1$ when $k$ is even.
\end{proposition}
\begin{proof}
First, we prove the case where $k$ is odd. Let $m$ be the maximum
size of strips of $c^*$ and suppose for a contradiction that $m >
2k$. We show vertices of this strip by
$v_i,v_{i+1},\dots,v_{i+m-1}$. Since $u_{i+k}$, $v_i$, and
$v_{i+2k}$ are selected by $\so(c^*)$, there is no need for
selection of $v_{i+k}$ in $\so(c^*)$, which is in contradiction
with Corollary~\ref{cor:new}. For even $k$, the process of the
proof is similar.
\end{proof}

The following lemma is trivial.

\begin{lemma}\label{lem3}
In a minimum cover $c^*$ of a generalized Petersen graph $P(n,k)$,
if a vertex $v\in V$ and its two adjacent vertices in $V$ are
selected in $c^*$, then there exists another minimum cover $c$ of
$P(n,k)$ that has the same selection from $V$ as $c^*$ with one
exception that it does not select $v$.
\end{lemma}

\begin{theorem}\label{thm3}
There is a minimum cover of $P(n,k)$ such that the maximum size of
its strips is at most $k+1$ when $k$ is odd, and $k+2$ when $k$ is
even.
\end{theorem}
\begin{proof}
Let $\mathcal{C'}(P(n,k))$ be the set of covers such that the
maximum size of their strips is smallest among all minimum
covers, and let $c^*$ be a cover in $\mathcal{C'}(P(n,k))$ with
the smallest number of maximum strips. Let $m$ be the maximum
size of strips of $c^*$. We claim that $c^*$ is a minimum cover
satisfying desired condition of the theorem.

First, we prove the theorem when $k$ is odd. By contradiction
suppose that $m > k+1$. Let the vertices of a strip with the
maximum size be $v_i,v_{i+1},\dots,v_{i+m-1}$. The vertex
$v_{i+2k}$ cannot be selected by $c^*$, otherwise if it is
selected, then the vertex $v_{i+k}$ satisfies the conditions of
Lemma~\ref{lem3}, and it means that there is another minimum
cover $c'$ that agrees with $c^*$ in choosing from $V$, except on
$v_{i+k}$. Therefore, in $c'$, the stated strip is divided into
two smaller ones and this contradicts the way we have chosen
$c^*$ among all covers. Similarly, $v_{i+2k+1}$ must not be
selected in $c^*$, also.

Now, we construct another minimum cover $c$ from $\so(c^*)$ in the
following way: $c$ selects the same vertices as $\so(c^*)$ except
$v_{i+2k}$ instead of $v_{i+k}$. Since $v_i$ and $u_{i+k}$ are
selected by $\so(c^*)$, $c$ is a cover. In addition, the strip with
size $m$ is destroyed in $c$ and we claim that the new strip has
size less than $m$. This contradicts the way we have chosen $c^*$.
To show the claim above, notice that $v_{i+m}$ and $v_{i+2k+1}$ are
not selected in $c$. Therefore, $v_{i+2k}$ is constructed a strip
with the size at most $2k-m$, which is less than $m$.

In the case that $k$ is even, the process of proof is similar.
\end{proof}

\section{A lower bound and some upper bounds}
In this section, we present some bounds for $\beta(P(n,k))$.
First, we introduce a lower bound.
\begin{proposition}\label{thm4}
If $n$ is odd then we have $\beta(P(n,k))\ge n+ \frac {(n,k)+1}{2}$.
\end{proposition}
\begin{proof}
First note that any vertex cover must select at least $\lceil
\frac{n}{2} \rceil = \frac{n+1}{2}$ vertices from $U$. Second, also
we note that the vertices of $V$ are in $(n,k)$ cycles of size
$\frac{n}{(n,k)}$, and a vertex cover must choose at least $\lceil
\frac{n}{(n,k)}/ 2 \rceil$ vertices from each of these cycles, so at
least $(n,k)\lceil \frac{n}{(n,k)}/ 2 \rceil= (n,k)
(\frac{n}{(n,k)}+1)/ 2 = \frac{n+(n,k)}{2}$ vertices of $V$ must be
included in any vertex cover. This shows that when $n$ is odd, a
vertex cover of $P(n,k)$ has at least $\frac{n+1}{2} +
\frac{n+(n,k)}{2} = n+\frac{(n,k)+1}{2}$ vertices.
\end{proof}
\begin{corollary}\label{cor:oddnbasic}
For all odd $n$, we have $\beta(P(n,k))\ge n + 1$.
\end{corollary}
Next, we find some upper bounds. The following simple bounds are
the results of Theorem~\ref{thm1}.
\begin{proposition}\label{cor1}
We have:
\begin{enumerate}
    \item[(a)]
    if $\frac{n}{(n,k)}$ is odd, then $\beta (P(n,k)) \leq
    n+\frac{n+(n,k)}{4}$, and
    \item[(b)]
    if $\frac{n}{(n,k)}$ is even, then $\beta (P(n,k))
    \leq n+\frac{n}{4}$.
\end{enumerate}
\end{proposition}
\begin{proof}
(a) We find an upper bound by introducing a semi-optimal cover of
$P(n,k)$. Since $V$ consists of $(n,k)$ pairwise disjoint cycles of
length $\frac{n}{(n,k)}$, we can cover all $V$-edges by selecting
alternatively $\lceil \frac{n}{(n,k)}/ 2 \rceil$ vertices from each
cycle. Now, let $c$ be a cover which consists of all vertices of $U$
and $\lceil \frac{n}{(n,k)}/ 2 \rceil$ vertices from each cycle of
$V$, chosen alternatively. By the Equation \eqref{socardi} in the
proof of Theorem~\ref{thm1}, we know that $\lvert \so(c)\rvert = n +
\frac{a(c)-b(c)}{2}$. Note that
\begin{equation*}
a(c) = (n,k) \left\lceil \frac{\frac{n}{(n,k)}}{2}\right\rceil
=(n,k)\frac{\frac{n}{(n,k)}+1}{2}= \frac{n+(n,k)}{2}.
\end{equation*}
Thus, we obtain
\begin{equation*}
\beta(P(n,k)) \le \lvert \so(c) \rvert=n+ \frac{a(c)-b(c)}{2} \leq
n+\frac{a(c)}{2} = n+\frac{n+(n,k)}{4}.
\end{equation*}

(b) The proof is similar to (a).
\end{proof}

We now present some more powerful upper bounds.

\begin{proposition}\label{thm5}
If both $n$ and $k$ are odd, then $\beta(P(n,k)) \le
n+\frac{k+1}{2}$.
\end{proposition}
\begin{proof}
Let $c$ select the vertices of $U$ with odd subscripts, and the
vertices of $V$ with even subscripts. It is easy to see that $c$
covers all $U$-edges and all spokes. Moreover, $c$ covers all
$V$-edges except the ones in which both endpoints have odd
subscripts. They are $\frac{k+1}{2}$ edges between the vertices
$v_{n-i}$ and $v_{n-i+k}$, for $i=0,2,4,\dots,k-1$. Now, if $c$
selects one endpoint of each of these edges, it becomes a cover of
size $n+\frac{k+1}{2}$.~\end{proof}

Now, we present a recursive upper bound in the following theorem. In
this theorem, we find a cover by applying a fixed pattern on
$m$-sectors of generalized Petersen graphs.

\begin{theorem}\label{thm6}
For a generalized Petersen graph $P(n,k)$, we have
\begin{enumerate}
    \item[(a)] if $m\mid n$, then $\beta(P(n,k))\le \frac{n}{m} \beta(P(m,r))$, and
    \item[(b)] if $m\nmid n$, then $\beta(P(n,k))\le \left\lfloor\frac{n}{m}\right
\rfloor \beta(P(m,r))+ 2k$
\end{enumerate}
where $m<k$ and $r\equiv k \pmod m$, $m > 2r >0$.
\end{theorem}
\begin{proof}
(a) In this case, we can partition $P(n,k)$ into $\frac{n}{m}$
sectors of size $m$. Now, we present a suitable pattern of selection
from an $m$-sector and apply this pattern on all $m$-sectors of
$P(n,k)$ such that selected vertices cover all edges of $P(n,k)$.
Since $r\equiv k\pmod m$, the $i$th vertex of an $m$-sector is
adjacent to $(i+r)$th vertex of one of the next $m$-sectors, where
$i+r$ is taken modulo $m$ in the range of $\{1,2,\dots,m\}$. By
considering this fact, it is not hard to see that we can find the
desired pattern by finding a cover of $P(m, r)$. By cutting $P(m,r)$
between the edges $(u_1,u_m)$ and $(v_1,v_m)$, we have an $m$-sector
with a pattern of selection which is defined by the cover of
$P(m,r)$.

Applying this pattern repeatedly on $m$-sectors yields a cover of
$P(n,k)$. Obviously, this selection covers all spokes and $U$-edges.
Furthermore, all $V$-edges must be covered, because otherwise, if an
edge $v_iv_{i+k}$ exists which is not covered, this means that the
edge $v_{i\pmod m}v_{(i+k)\pmod m}$ or equivalently $v_jv_{j+r}$ is
not covered in $P(m,r)$, where $j\equiv i \pmod m$, which is a
contradiction.

(b) In this case, we can apply the idea of previous case with some
modifications. We partition $P(n,k)$ from $u_1$ and $v_1$ into
$\lfloor\frac{n}{m}\rfloor$ consecutive $m$-sectors. Then, we apply
the pattern obtained from the cover of $P(m, r)$ to these sectors.
Let $r' \equiv n \pmod m$. Therefore, an $r'$-sector is left. For
covering $P(n,k)$ completely, it is not hard to see that it suffices
to choose $u_{n-r'+1}, \dots, u_n$ and
$v_{n-k+1},\dots,v_n,\dots,v_{n+k-r'}$ in addition to other
vertices. Now, we reach the following upper bound for $m < k$:
\begin{align*}
\beta(P(n,k))&\le\beta(P(m,r))\left\lfloor\cfrac{n}{m}\right
\rfloor+r'+(n+k-r'-(n-k+1)+1)\\
&\le\beta(P(m,r))\left\lfloor\cfrac{n}{m}\right \rfloor+2k.
\end{align*}
\end{proof}

\begin{remarka}
As it is clear from the proof of the Theorem~$\ref{thm6}$(b), the
bound can be improved, because to be sure about covering of edges,
we blindly selected $2k$ vertices (that some of them might have been
selected before in their $m$-sectors).
\end{remarka}
\begin{corollary}\label{cor3}
For all even $k$, we have
\begin{enumerate}
    \item[(a)] if $k-1\mid n$, then $\beta(P(n,k))\le
n+\cfrac{n}{k-1}$
    \item[(b)] if $k-1\nmid n$, then $\beta(P(n,k))\le
n+\left\lfloor\cfrac{n}{k-1}\right \rfloor+2k$.
\end{enumerate}
\end{corollary}
\begin{proof}
(a) It is enough to let $m=k-1$ in Theorem~\ref{thm6}. In this case,
we obtain
\begin{align*}
\beta(P(n,k))&\le\cfrac{n}{m}\beta(P(m,r))\\
&=\cfrac{n}{k-1}\beta(P(k-1,1))\\
&=\cfrac{n}{k-1}k = n + \cfrac{n}{k-1}.
\end{align*}
(b) The proof is similar to part (a).
\end{proof}

\section{Some exact values for $\beta(P(n,k))$}
In this section, we introduce exact values of $\beta(P(n,k))$ for
some $n$ and $k$.
\begin{lemma}\label{lem4}
$P(n,k)$ is a bipartite graph if and only if $n$ is even and $k$ is
odd.
\end{lemma}
\begin{proof}
For odd $n$, all vertices of $U$ form an odd cycle. In addition,
for $n$ and $k$ both even, the cycle $u_1 v_1 v_{k+1} u_{k+1}
u_{k} u_{k-1} ... u_2 u_1$ is an odd cycle, $C_{k+3}$. For even
$n$ and odd $k$, let $X=\{ u_i, v_{i+1} | i$ odd $ \}$ and $Y=\{
u_i, v_{i+1}| i$ even $\}$. It is easy to see that this is a
bipartition of vertices of $P(n,k)$.
\end{proof}

\begin{proposition}\label{prp1}
$\beta (P(n,k))=n$ if and only if $n$ is even and $k$ is odd.
\end{proposition}
\begin{proof}
If $\beta (P(n,k)) = n$ then by Corollary~\ref{cor:oddnbasic}, $n$
is not odd. To cover $U$-edges and spokes, any minimum cover $c^*$
must alternatively select $\frac{n}{2}$ vertices from $U$ and
$V$, say without loss of generality with beginning from $u_1$ and
$v_2$ respectively. Therefore, when $k$ is even, the edge
$v_1v_{k+1}$ is not covered by $c^*$, which is a contradiction.
This completes proof of sufficiency. Conversely, for even $n$ and
odd $k$, by Lemma~\ref{lem4}, $P(n,k)$ is a bipartite graph, thus
for an arbitrary bipartition, the smaller part is a cover of size
at most $n$. Also notice $P(n,k)$ has a matching of size $n$, for
example the spokes, thus the size of any cover must be at least
$n$. Therefore, we have $\beta (P(n,k))=n$.
\end{proof}

\begin{theorem}\label{thm:nsn+1}
$\beta(P(n,k))= n + 1$ if and only if $n$ is odd and $k = 1$, or
$(n,k)=(5,2)$.
\end{theorem}
\begin{proof}
\textit{Sufficiency.} Trivially $\beta(P(5,2))=6=5+1$. Now suppose
$n$ is odd and $k = 1$. By Corollary~\ref{cor:oddnbasic}, we have
$\beta(P(n,k)) \ge n+1$. Consider a cover that selected vertices
$\{u_i|i$~odd~$\}$ and $\{v_i|i=1$ or $i$ even $\}$. Since this is a
cover of size $n+1$, we obtain $\beta(P(n,k)) = n+1$.

\textit{\it Necessity.} Suppose $\beta(P(n,k)) = n+1$. Let $c^*$ be
an arbitrary minimum cover of $P(n,k)$. If $n = 2k + 1$, then
$\beta(P(2k+1,k))=2k+1+\lceil\frac{2k+1}{5}\rceil$
(see~\cite{behzad06}) and therefore, $\lceil\frac{2k+1}{5}\rceil =
1$; and so the only solutions in this case are $(n,k)=(3,1)$ and
$(n,k)=(5,2)$. Now suppose $n> 2k+1$. By considering
Equation~\eqref{socardi} in the proof of Theorem~\ref{thm1} for
$c^*$, we have $a(c^*)-b(c^*)=2$. Therefore, the size of each strip
of $c^*$ must be less than $4$. Thus, it is easy to check that all
strips have size $1$ except exactly one strip with size $2$ or $3$.

We claim that $c^*$ must select at least one vertex from every
pair of consecutive vertices of $V$. Otherwise, without loss of
generality, say the two consecutive unselected vertices are
$v_{k+1}$ and $v_{k+2}$. Clearly, $v_1, v_2, v_{2k+1}$, and
$v_{2k+2}$ are selected by $c^*$, and they form two strips of
size more than one or one strip of size more than three, a
contradiction. Therefore, without loss of generality the strip of
size greater than one is $\{v_n, v_1, v_2\}$ when $n$ is even,
and is $\{v_1, v_2\}$ when $n$ is odd, and $c^*$ must select
alternatively from vertices of $V$ beginning with $v_2$.

If $k$ is even, then both endpoints of edges $v_1v_{k+1},
v_3v_{k+3}, \dots, v_{k+1}v_{2k+1}$ have odd subscripts. Note
that $c^*$ has selected vertices with even subscripts and $v_1$.
So it can cover only one of these $\frac{k}{2}+1$ edges. This
leaves $\frac{k}{2}$ uncovered edges, a contradiction. If $k$ is
odd (which by Proposition~\ref{prp1} implies that $n$ is also
odd), then both endpoints of edges between the vertices $v_{n-i}$
and $v_{n-i+k}$, for $i=0,2,4,\dots,k-1$ have odd subscripts.
Again, since that $c^*$ has selected vertices with even
subscripts and $v_1$, it can cover only one of these
$\lceil\frac{k}{2}\rceil$ edges. Hence, we have $k=1$. This
completes proof of necessity.
\end{proof}

By Proposition~\ref{prp1} and Theorem~\ref{thm:nsn+1}, the
following is immediate.
\begin{corollary}\label{cor:>n+1}
$\beta(P(n,k))\ge n+2$ if and only if none of the following
conditions holds:
\begin{enumerate}
    \item [(a)] $n$ is even and $k$ is odd,
    \item [(b)] $n$ is odd and $k = 1$,
    \item [(c)] $n = 5$ and $k = 2$.
\end{enumerate}
\end{corollary}

\begin{proposition}\label{divisableprop}
For $n$ and $k$ both odd, where $k\mid n$,
$\beta(P(n,k))=n+\frac{k+1}{2}$.
\end{proposition}
\begin{proof}
If $k|n$, $(n,k)=k$. Thus, by Propositions~\ref{thm4} and~\ref{thm5}
, we have
\begin{equation*}
n+\frac{(n,k)+1}{2} = n+\frac{k+1}{2}\leq \beta(P(n,k))\leq
n+\frac{k+1}{2}. 
\end{equation*}
\end{proof}
\\
By the results above, we have the following precise values.

\begin{proposition}
\ \ $\beta (P(n,1)) =
\begin{cases}
    n        &\mbox{if n is even} \\
    n + 1    &\mbox{if n is odd}.
 \end{cases}$
\end{proposition}

And we recall the following result of~\cite{behzad06}:

\begin{proposition}
\ \ $\beta (P(n,2))= n + \lceil \frac{n}{5} \rceil$.
\end{proposition}

\begin{proposition}\label{thm9}
\ \ $\beta (P(n,3)) =
\begin{cases}
    n        &\mbox{if n is even} \\
    n + 2    &\mbox{if n is odd}.
 \end{cases}$
\end{proposition}
\begin{proof}
If $n$ is even, it is implied by Proposition~\ref{prp1}. If $n$ is
odd, by Corollary~\ref{cor:>n+1}, $\beta(P(n,3)) \ge n+2$ and by
Proposition~\ref{thm5}, $\beta(P(n,3)) \le n+2$. Thus, we obtain
$\beta(P(n,3))= n + 2$.
\end{proof}

\section{Concluding remarks}
In two previous sections, we presented bounds and some exact
values for the size of minimum cover of generalized Petersen
graphs. With careful inspection of these results, we reach to the
following proposition.
\begin{proposition}\label{prp:conj}
For large enough $n$ and for a fixed $k \ne 4$, we have
$\beta(P(n,k))\le n + \lceil\frac{n}{5}\rceil + O(1)$.
\end{proposition}
\begin{proof}
This can be proved from the following facts:
\begin{enumerate}
    \item for an odd $k$, by Propositions~\ref{thm5} and~\ref{prp1},
    the statement is correct.
    \item for $k = 2$, according to~\cite{behzad06}, the statement is held.
    \item for an even $k > 4$, by Corollary~\ref{cor3}, the statement
    is true.
\end{enumerate}\end{proof}

Due to Proposition~\ref{prp:conj} and our observations for
generalized Petersen graphs for small $n$, we conjecture the
following.
\begin{conjecture}
For all $n$ and $k$, $\beta(P(n,k))\leq n+\lceil\frac{n}{5}\rceil$.
\end{conjecture}
This conjecture is checked to be true for all $2k<n$, where $n\leq
35$.

\begin{question}It seems that by using ideas similar to the ones used in
the proof of Theorems \ref{thm2} and \ref{thm3}, an algorithm with
polynomial complexity for finding minimum cover of generalized
Petersen graphs can be found. But we have not still  found one.
\end{question}

\section*{Acknowledgements}
This work was partly done while the third author was spending his
sabbatical leave in  the Institute for Advanced Studies in Basic
Sciences (IASBS), Zanjan. We would like to thank the Department
of Mathematics at IASBS for their warm and generous hospitality
and support.


\end{document}